\documentclass[english]{article}
\usepackage[T1]{fontenc}
\usepackage[latin9]{inputenc}
\usepackage{geometry}
\usepackage{array}
\usepackage{float}
\usepackage{amsthm}
\usepackage{amsmath}
\usepackage{graphicx}

\makeatletter

\providecommand{\tabularnewline}{\\}

\theoremstyle{plain}
\newtheorem{thm}{\protect\theoremname}
  \theoremstyle{definition}
  \newtheorem{defn}[thm]{\protect\definitionname}
  \theoremstyle{remark}
  \newtheorem{rem}[thm]{\protect\remarkname}

\usepackage{times}
\usepackage{txfonts}

\usepackage{babel}
\providecommand{\definitionname}{Definition}
  \providecommand{\remarkname}{Remark}
\providecommand{\theoremname}{Theorem}

\usepackage{babel}
\providecommand{\definitionname}{Definition}
  \providecommand{\remarkname}{Remark}
\providecommand{\theoremname}{Theorem}

\makeatother

\usepackage{babel}
  \providecommand{\definitionname}{Definition}
  \providecommand{\remarkname}{Remark}
\providecommand{\theoremname}{Theorem}

\begin{document}

\title{Core-satellite Graphs: Clustering, Assortativity and Spectral Properties}

\author{Ernesto Estrada\footnotemark[2] \and Michele Benzi\footnotemark[3]}

\maketitle
\footnotetext[2]{Department of Mathematics \& Statistics, University
of Strathclyde, 26 Richmond Street, Glasgow G1 1HX, UK (ernesto.estrada@strath.ac.uk).}
\footnotetext[3]{Department of Mathematics and Computer Science,
Emory University, Atlanta, Georgia 30322, USA (benzi@mathcs.emory.edu).
}

\begin{abstract}
Core-satellite graphs (sometimes referred to as generalized friendship
graphs) are an interesting class of graphs that generalize many well
known types of graphs. In this paper we show that two popular clustering
measures, the average Watts-Strogatz clustering coefficient and the
transitivity index, diverge when the graph size increases. We also
show that these graphs are disassortative. In addition, we completely
describe the spectrum of the adjacency and Laplacian matrices associated
with core-satellite graphs. Finally, we introduce the class of generalized
core-satellite graphs and analyze their clustering, assortativity, and
spectral properties. 
\end{abstract}

\section{Introduction}

The availability of data about large real-world networked systems---commonly
known as complex networks---has demanded the development of several
graph-theoretic and algebraic methods to study the structure and dynamical
properties of such usually giant graphs \cite{EstradaBook,NewmanReview,LucianoReview}.
In a seminal paper, Watts and Strogatz \cite{WS} introduced one of
such graph-theoretic indices to characterize the transitivity of relations
in complex networks. The so-called clustering coefficient represents
the ratio of the number of triangles in which the corresponding node
takes place to the the number of potential triangles involving that
node (see further for a formal definition). The clustering coefficient
of a node is bounded between zero and one, with values close to zero
indicating that the relative number of transitive relations involving
that node is low. On the other hand, a clustering coefficient close
to one indicates that this node is involved in as many transitive
relations as possible. When studying complex real-world networks it
is very common to report the average Watts-Strogatz (WS) clustering
coefficient $\bar{C}$ as a characterization of how globally clustered
a network is \cite{EstradaBook,NewmanReview}. Such idea, however,
was not new as reflected by the fact that Luce and Perry \cite{Luce and Perry}
had proposed 50 years earlier an index to account for the network
transitivity, given by the total number of triangles in the graph
divided by the total number of triads existing in the graph. Such
index, hereafter called the graph transitivity $C$, was then rediscovered
by Newman \cite{Newman clustering} in the context of complex networks.
Here again this index is bounded between zero and one, with small
values indicating poor transitivity and values close to one indicating
a large one (see also \cite{Wasserman_Faust}).

It was first noticed by Bollob\'as \cite{Bollobas} and then by Estrada
\cite{EstradaBook} that there are certain graphs for which the two
clustering coefficients diverge. That is, there are classes of graphs
for which the WS average clustering coefficient tends to one while
the graph transitivity tends to zero as the size of the graphs grows
to infinity. The two families indentified by Bollob\'as \cite{Bollobas}
and by Estrada \cite{EstradaBook} are illustrated in Figure \ref{examples}
and they correspond to the so-called \textit{friendship} (or Dutch
windmill or $n$-fan) \textit{graphs} \cite{friendship theorem} and
\textit{agave graphs} \cite{Agave graphs}, respectively. The friendship
graphs are formed by glueing together $\eta$ copies of a triangle
in a common vertex. The agave graphs are formed by connecting $s$
disjoint nodes to both nodes of a complete graph $K_{2}$.

The friendship graphs are members of a larger family of graphs known
as the \textit{windmill graphs} \cite{windmill graphs,windmill spectra_1,windmill spectra_2,windmill spectra_3,windmill spectra_4}.
A windmill graph $W\left(\eta,s\right)$ consists of $\eta$ copies
of the complete graph $K_{s}$ \cite{windmill graphs} with every
node connected to a common one (see Figure \ref{windmill graphs}).
In a recent paper Estrada \cite{clustering divergence} has proved
that the divergence of the two clustering coefficient indices also
takes place for windmill graphs when the number of nodes tends to
infinity. The agave graphs belong to the general class of complete
split graphs, which are the graphs consisting of a central clique
and a disjoint set of nodes which are connected to all the nodes of
the clique \cite{split graphs}.

\begin{figure}[H]
\begin{centering}
\includegraphics[width=0.75\textwidth]{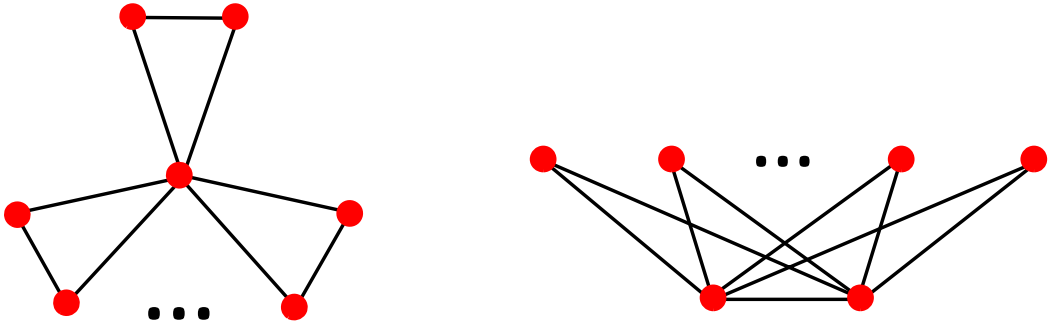} 
\par\end{centering}

\protect\protect\protect\caption{Examples of graphs for which the average WS clustering coefficient
and the graph transitivity diverge as the size of the graphs goes
to infinity.}

\label{examples} 
\end{figure}

It is important to note that this divergence between the two clustering
coefficients also takes place in certain real-world networks \cite{clustering divergence}.
This phenomenon has passed inadverted among the plethora of papers
that deals with clustering in real-world networks. In Table \ref{Table1}
we illustrate some results collected from the literature in which
the divergence of the two clustering coefficients is observed. The
dataset is selected to intentionally display the heterogeneity of
classes of real-world networks which display such phenomenon. They
correspond to the interchange of emails among the employees of the
ENRON corporation, a CAIDA version of the Internet as Autonomous System,
a metabolic network, a social network of Hollywood actors and a network
of coauthors in Biology.

\begin{table}[H]
\begin{centering}
\begin{tabular}{>{\centering}p{2cm}>{\centering}p{1.7cm}>{\centering}p{1.8cm}>{\centering}p{1cm}>{\centering}p{1cm}}
\hline 
 & $n$  & $m$  & $\bar{C}$  & $C$\tabularnewline
\hline 
E-mail  & 36 692  & 183 831  & 0.50  & 0.09\tabularnewline
Internet  & 26 475  & 106 762  & 0.21  & 0.007\tabularnewline
metabolic  & 765  & 3 686  & 0.67  & 0.09\tabularnewline
Actors  & 449 913  & 25 516 482  & 0.78  & 0.20\tabularnewline
Coauthors  & 1 520 251  & 11 803 064  & 0.60  & 0.09\tabularnewline
\hline 
\end{tabular}
\par\end{centering}

\protect\protect\protect\caption{Some real-world networks collected from the literature displaying
divergence of the average WS clustering coefficient and the graph
transitivity. Here $n$ and $m$ denote the number of nodes and edges
in the various graphs. }

\label{Table1} 
\end{table}

In this work we study a class of graphs which accounts for the clustering
divergence phenomenon in networks with degree disassortativity. These
graphs, which we call {\em core-satellite graphs}, are formed by
a central clique (the core) connected to several other cliques (the
satellites) in such a way that they generalize both the complete split
and the windmill graphs. We prove here that the clustering coefficients
of these graphs diverge and that they are always disassortative. We
also characterize completely the eigenstructure of the adjacency and
the Laplacian matrices of these graphs, and comment on the asymptotic
behavior of such quantities as the infection threshold and synchronizability
index of these graphs. Finally, we provide evidence that supports
the idea of using the core-satellite graphs as models of certain classes
of real-world networks.

\section{Preliminaries}

Here we always consider \textit{simple} \textit{undirected and connected
graphs} $G=(V,E)$ formed by a set of $n$ nodes (vertices) $V$ and
a set of $m$ edges $E=\{(u,v)\,|\,u,v\in V\}$ between the nodes.
Let us now recall an important graph operation which will be helpful
in this work. 
\begin{defn}
The \textit{join} (or \textit{complete product}) $G_{1}\nabla G_{2}$
of graphs $G_{1}$ and $G_{2}$ is the graph obtained from $G_{1}\cup G_{2}$
by joining every vertex of $G_{1}$ with every vertex of $G_{2}$. 
\end{defn}
Let us now define some specific classes of graphs which will be considered
in this work. As usual, we denote with $\bar{G}$ the complement of
a graph $G$.
\begin{defn}
The complete split graph is $\Sigma\left(a,b\right)\cong K_{a}\nabla\bar{K}_{b}$.
The special cases $\Sigma\left(1,b\right)\cong K_{2}\nabla\bar{K}_{b}$
and $\Sigma\left(2,b\right)\cong K_{2}\nabla\bar{K}_{b}$ are known
as the star and agave graphs, respectively. 
\end{defn}
In Figure \ref{complete split } we show some examples of complete
split graphs.

\begin{figure}[H]
\begin{centering}
\includegraphics[width=0.75\textwidth]{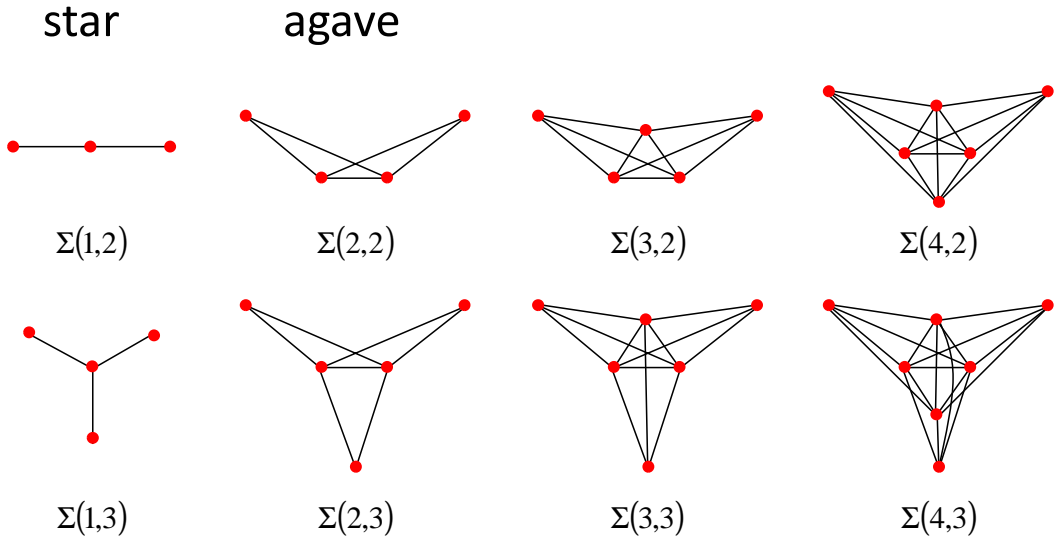} 
\par\end{centering}

\protect\protect\protect\caption{Illustration of some small complete split graphs including the classes
of star and agave graphs.}

\label{complete split } 
\end{figure}

\begin{defn}
The windmill graph is $W\left(\eta,a\right)\cong K_{1}\nabla\left(\eta K_{a}\right)$.
These are the graphs consisting of $\eta$ copies of $K_{a}$ meeting
in a common vertex. The special cases $W\left(\eta,2\right)\cong K_{1}\nabla\left(\eta K_{2}\right)$
are known as the friendship graphs. 
\end{defn}
In Figure \ref{windmill graphs} we show some examples of windmill
graphs.

\begin{figure}[H]
\begin{centering}
\includegraphics[width=0.7\textwidth]{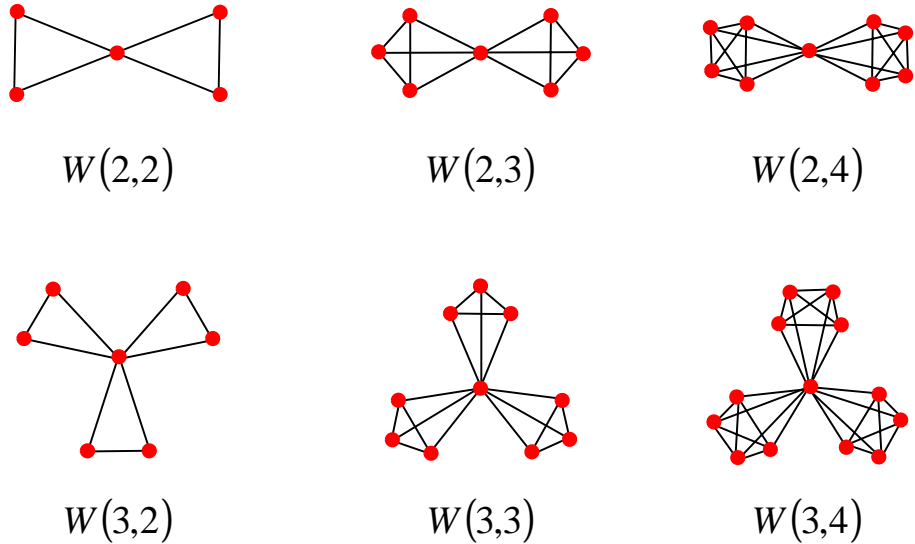} 
\par\end{centering}

\protect\protect\protect\caption{Illustration of some small windmill graphs including the classes of
friendship graphs.}

\label{windmill graphs} 
\end{figure}

Let us now define a few graph-theoretic invariants that will be studied
in this work. The so-called \textit{Watts-Strogatz clustering coefficient}
of a node $i$, which quantifies the degree of transitivity of local
relations in a graph is defined as \cite{WS}:

\begin{equation}
C_{i}=\frac{2t_{i}}{k_{i}(k_{i}-1)}
\end{equation}
where $t_{i}$ is the the number of triangles in which the node $i$
participates. Taking the mean of these values as $i$ varies among
all the nodes in $G$, one gets the \textit{average WS clustering
coefficient} of the network $G$:

\begin{equation}
\overline{C}=\frac{1}{n}\sum_{u=1}^{n}C_{u}.
\end{equation}

The so-called graph transitivity is defined as \cite{Newman clustering,Wasserman_Faust}

\begin{equation}
C=\frac{3t}{P_{2}},
\end{equation}

where $t$ is the total number of triangles and $P_{2}=\sum_{i=1}^{n}k_{i}(k_{i}-1)/2$.

Another important network parameter is the degree assortativity coefficient
which measures the tendency of high degree nodes to be connected to
other high degree nodes (assortativity) or to low degree ones (disassortativity).
The assortativity coefficient is mathematically expressed as \cite{Newman assortativity}:

\begin{equation}
r=\dfrac{m^{-1}\sum_{i,j\in E}k_{i}k_{j}-\dfrac{1}{4m^{2}}\left(\sum_{i,j\in E}k_{i}+k_{j}\right)^{2}}{\dfrac{1}{2m}\left(\sum_{i,j\in E}k_{i}^{2}+k_{j}^{2}\right)-\dfrac{1}{4m^{2}}\left(\sum_{i,j\in E}k_{i}+k_{j}\right)^{2}}.
\end{equation}

Let $A=\left(a_{uv}\right)_{n\times n}$ be the \textit{adjacency
matrix} of the graph. We denote by $\lambda_{1}>\lambda_{2}\geq\cdots\geq\lambda_{n}$
the eigenvalues of $A$. The dominant eigenvalue $\lambda_{1}$ is
also referred to as the Perron eigenvalue or spectral radius of $A$,
denoted $\rho(A)$. 
Let $\varDelta$ be the diagonal matrix of the vertex degrees of the graph.
Then, the Laplacian matrix of the graph is defined by $L=\varDelta-A$.
We denote by $\mu_{1}\ge\mu_{2}\geq\cdots\geq\mu_{n-1}>\mu_{n}=0$
the eigenvalues of $L$ . The eigenvalue $\mu_{n-1}$ is known as
the \textit{algebraic connectivity} of the graph and the eigenvector
associated to it is known as the \textit{Fiedler vector} \cite{Fiedler}.
For reviews about spectral properties of graphs the reader is directed
to the classic works \cite{Spectra of graphs_1,Spectra of graphs_2}.

\section{Core-satellite graphs}

The main paradigm for defining the core-satellite graphs is the following.
We consider a group of central nodes in a network which are connected
among them. Then, there are a few cliques of the same size which are
connected to the central core but which are not connected among them.
Formally, the core-satellite graphs are defined below. 
\begin{defn}
Let $c\geq1$, $s\geq1$ and $\eta\geq2$. The core-satellite graph is
$\Theta\left(c,s,\eta\right)\cong K_{c}\nabla\left(\eta K_{s}\right)$.
That is, they are the graphs consisting of $\eta$ copies of $K_{s}$
(the satellites) meeting in a common clique $K_{c}$ (the core), see
Fig. \ref{core-satellite}.

\begin{figure}[H]
\begin{centering}
\includegraphics[width=0.75\textwidth]{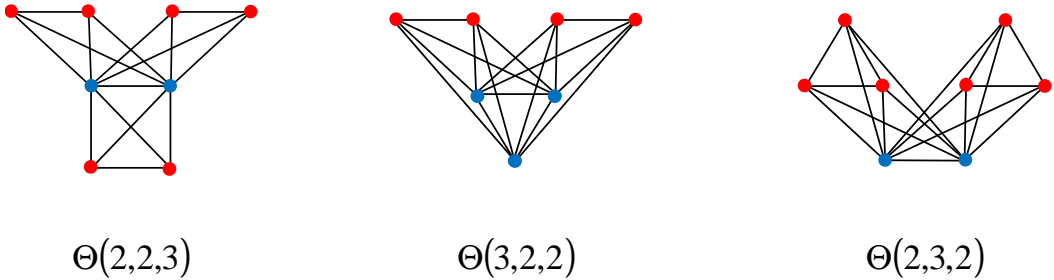} 
\par\end{centering}

\protect\protect\protect\caption{Examples of core-satellite graphs. Nodes in the core are drawn in
blue and those in the satellites in red.}

\label{core-satellite} 
\end{figure}
\end{defn}
\begin{rem}
The core-satellite graph generalizes the following classes of graphs:\end{rem}
\begin{itemize}
\item star graphs: $\Theta\left(1,1,\eta\right)$; 
\item agave graphs: $\Theta\left(2,1,\eta\right)$; 
\item complete split graphs: $\Theta\left(c,1,\eta\right)$; 
\item friendship graphs: $\Theta\left(1,2,\eta\right)$; 
\item windmill graphs: $\Theta\left(1,s,\eta\right)$. 
\end{itemize}
The core-satellite graphs were previously defined in the literature
in the context of extremal graph theory. They are named \textit{generalized
friendship graphs} in the works of Erd\H{o}s et al. \cite{generalized friendship_1},
Chen et al. \cite{generalized friendship_2} and Faudree et al. \cite{generalized friendship_3}.
However, the same name appears in connection to at least another class
of graphs. The generalized friendship graph is defined by Ahmad et
al. \cite{generalized friendship_4}, Arumugan and Nalliah \cite{generalized friendship_5},
Shi et al. \cite{generalized friendship_6} and by Jahari and Alikhani
\cite{generalized friendship_7} as the graph consisting of $r$ cycles
of orders $n_{1}\geq n_{2}\geq\cdots\geq n_{r}$ having a common vertex.
These graphs are also known as flowers. Thus, to avoid any confusion
we propose here the more intuitive name of \textit{core-satellite}
for these graphs.

\section{General properties of core-satellite graphs} \label{sec:Gen}

Let $\Theta\left(c,s,\eta\right)$ be a core-satellite graph and let
us designate any node in the core as $i$ and any node in a satellite
as $j$. Then, $k_{i}=c+\eta s-1$ and $k_{j}=c+s-1$, also $n=c+\eta s$
and $m=\left(\begin{array}{c}
c+p\\
2
\end{array}\right)-\left(\eta-1\right)\left(\begin{array}{c}
c\\
2
\end{array}\right)$. We now state our first result. 
\begin{thm}
\label{thm:clustering}Let $\Theta\left(c,s,\eta\right)$ be a core-satellite
graph. Then, for given values of $c$ and $s$ the average Watts-Strogatz
and the transitivity coefficients diverge when the number of cliques
$\eta$ tends to infinity:

\begin{equation}
\underset{\eta\rightarrow\infty}{\lim}\bar{C}=1,
\end{equation}

\begin{equation}
\underset{\eta\rightarrow\infty}{\lim}C=0.
\end{equation}
\end{thm}
\begin{proof}
First, we obtain an expression for the Watts-Strogatz clustering coefficient
of windmill graphs. The Watts-Strogatz clustering coefficient of any
node in the core clique of $\Theta\left(c,s,\eta\right)$ is
\begin{equation}
C_{i}=\dfrac{\eta\left(c+s-1\right)\left(c+s-2\right)-\left(c-1\right)\left(c-2\right)\left(\eta-1\right)}{\left(\eta s+c-1\right)\left(\eta s+c-2\right)},\label{eq:node clustering}
\end{equation}
and the rest of the nodes have $C_{j}=1$. Thus, $\bar{C}=\left(cC_{i}+\eta s\right)/n$,
which gives

\begin{equation}
\bar{C}=1-\dfrac{cs^{2}\eta^{2}}{n\left(n-1\right)\left(n-2\right)}=1-\dfrac{c\left(n-c\right)^{2}}{n\left(n-1\right)\left(n-2\right)}
\end{equation}

Now we consider the transitivity index of a windmill graph $\Theta\left(c,s,\eta\right)$.
The total number of triangles in $W\left(\eta,\kappa\right)$ is
\begin{equation}
t=\eta\left(\begin{array}{c}
c+s\\
3
\end{array}\right)+\left(\eta-1\right)\left(\begin{array}{c}
c\\
3
\end{array}\right),
\end{equation}
and the number of 2-paths is
\begin{equation}
P_{2}=c\left(\begin{array}{c}
n-1\\
2
\end{array}\right)+\eta s\left(\begin{array}{c}
c+s\\
2
\end{array}\right).
\end{equation}
Thus, after substitution we get
\begin{equation}
C=\dfrac{\eta\left(c+s\right)\left(c+s-1\right)\left(c+s-2\right)-c\left(\eta-1\right)\left(c-1\right)\left(c-2\right)}{\eta s\left(c+s-1\right)\left(c+s-2\right)+c\left(c+\eta s-1\right)\left(c+\eta s-2\right)}.
\end{equation}
Obviously, $\underset{\eta\rightarrow\infty}{\lim}\bar{C}=1$ and
$\underset{\eta\rightarrow\infty}{\lim}C=0$
for any given values of $c$ and $s$, which proves the result.\hfill{} 
\end{proof}

Next, we prove a result related to the way in wich nodes connect to
each other in core-satellite graphs.
\begin{thm}
Let $\Theta\left(c,s,\eta\right)$ be a core-satellite graph. Then,
for any values of $c\geq1$, $s\geq1$, and $\eta\geq2$ the core-satellite
graphs are disassortative, i.e., $r<0.$\end{thm}
\begin{proof}
We first need a result previously obtained by Estrada \cite{Estrada assortativity}
that expresses the assortativity coefficient in terms of subgraphs
of the corresponding graph,

\begin{equation}
r=\dfrac{P_{2}\left(P_{3/2}+C-P_{2/1}\right)}{3S_{1,3}-P_{2}\left(P_{2/1}-1\right)},\label{eq:assortativity}
\end{equation}
where $P_{s/t}$ represents the ratio of paths of length $s$ to paths
of length $t$, $S_{1,3}$ is the number of star subgraphs with a
central node and 3 pendant nodes, and $C$ is the transitivity index.
It has been previously proved in \cite{Estrada assortativity} that
the denominator of (\ref{eq:assortativity}) is positive. Thus,
the sign of the assortativity coefficient depends on the sign of $P_{3/2}+C-P_{2/1}$.
It is straighforward to realize that in a core-satellite graph the
number of paths of length 3 is zero. Consequently, the sign of $r$
depends only on the sign of $C-P_{2/1}$. We remind that $0\leq C\leq1$.
Thus, let us consider the difference $P_{2}-P_{1}$

\begin{equation}
P_{2}-P_{1}=\dfrac{ck_{i}}{2}\left(k_{i}-1\right)+\dfrac{\eta sk_{j}}{2}\left(k_{j}-2\right).
\end{equation}

Because $\eta\geq2$ then $k_{i}\geq2$ (notice that if $k_{i}=1$
the resulting graph is just $K_{2})$. If $k_{j}=1$ the core-satellite
graph corresponds to the star graph, which is a tree, and consequently
has $C=0$. Thus, the graph is diassortative. Let $k_{j}\geq2$, then
because $\eta\geq2$ we have that $k_{i}\geq4$, which implies that
$P_{2}-P_{1}>1$, i.e., $P_{2/1}>1$ and consequently $r<0$, which
proves the result. 
\end{proof}

\section{Spectral properties of core-satellite graphs}

In this section we give a full description of the spectral properties
(eigenvalues and eigenvectors) of core-satellite graphs, and we investigate
extensions to more general types of graphs. We begin by proving the
following result. 
\begin{thm}
\label{thm:spectrum} The spectrum of the core-satellite graph $\Theta(c,s,\eta)$
consists of: 
\begin{enumerate}
\item The eigenvalue $\lambda=-1$ with multiplicity $c+\eta(s-1)-1$; 
\item The eigenvalue $\lambda=s-1$ with multiplicity $\eta-1$; 
\item The eigenvalues $\lambda_{\pm}$ given by the roots of the quadratic
equation 
\begin{equation}
\lambda^{2}-(c+s-2)\lambda+(c-1)(s-1)-\eta cs=0.\label{eq:quadratic}
\end{equation}

\end{enumerate}
\end{thm}
\begin{proof}
For any positive integer $p$, we use ${\bf 1}_{p}$ to denote the
column vector with all entries equal to 1. Let $A(K_{c})$ and $A(K_{s})$
denote, respectively, the adjacency matrices of the complete graphs
with $c$ and $s$ nodes. Note that $A(K_{c})={\bf 1}_{c}{\bf 1}_{c}^{T}-I_{c}$,
$A(K_{s})={\bf 1}_{s}{\bf 1}_{s}^{T}-I_{s}$ with $I_{c}$ and $I_{s}$
standing for the $c\times c$ and $s\times s$ identity matrices.
With an obvious ordering of the nodes, the adjacency matrix of $\Theta(c,s,\eta)$
can be written as the block matrix 
\begin{equation}
A=\left[\begin{array}{c|ccc}
A(K_{c}) & {\bf 1}_{c}{\bf 1}_{s}^{T} & \cdots & {\bf 1}_{c}{\bf 1}_{s}^{T}\\[2pt]
\hline \\[-8pt]
{\bf 1}_{s}{\bf 1}_{c}^{T} & A(K_{s})\\
\vdots &  & \ddots\\[2pt]
{\bf 1}_{s}{\bf 1}_{c}^{T} &  &  & A(K_{s})
\end{array}\right],\label{eq:adj}
\end{equation}
with $\eta$ copies of $A(K_{s})$ appearing on the block diagonal.
Now let the vector 
\begin{equation}
{\bf x}=\left[\begin{array}{c}
{\bf x}_{1}\\[2pt]
\hline \\[-8pt]
{\bf 0}\\
\vdots\\[2pt]
{\bf 0}
\end{array}\right],\qquad{\bf x}_{1}\ne{\bf 0},\label{eq:ev1}
\end{equation}
be partitioned conformally to $A$. Then the equation $A{\bf x}=\lambda{\bf x}$
becomes 
\[
\left[\begin{array}{c}
A(K_{c}){\bf x}_{1}\\[2pt]
\hline \\[-8pt]
({\bf 1}_{c}^{T}{\bf x}_{1}){\bf 1}_{s}\\
\vdots\\[2pt]
({\bf 1}_{c}^{T}{\bf x}_{1}){\bf 1}_{s}
\end{array}\right]=\left[\begin{array}{c}
\lambda{\bf x}_{1}\\[2pt]
\hline \\[-8pt]
{\bf 0}\\
\vdots\\[2pt]
{\bf 0}
\end{array}\right].
\]
Clearly, this is equivalent to $A(K_{c}){\bf x}_{1}=\lambda{\bf x}_{1}$,
and ${\bf x}_{1}\perp{\bf 1}_{c}$. Hence, $\lambda=-1$ and there
are $c-1$ linearly independent eigenvectors of the form (\ref{eq:ev1}).
Next, consider vectors of the form 
\begin{equation}
{\bf x}=\left[\begin{array}{c}
{\bf 0}\\[2pt]
\hline \\[-8pt]
{\bf 0}\\
\vdots\\[2pt]
{\bf 0}\\
{\bf x}_{i}\\[2pt]
{\bf 0}\\
\vdots\\[2pt]
{\bf 0}
\end{array}\right],\qquad{\bf x}_{i}\ne{\bf 0},\qquad i=1,\ldots,\eta,\label{eq:ev2}
\end{equation}
where ${\bf x}_{i}$ corresponds to the $(i+1)$th diagonal block
of $A$. Then the equation $A{\bf x}=\lambda{\bf x}$ becomes 
\[
\left[\begin{array}{c}
({\bf 1}_{s}^{T}{\bf x}_{i}){\bf 1}_{c}\\[2pt]
\hline \\[-8pt]
{\bf 0}\\
\vdots\\[2pt]
{\bf 0}\\
A(K_{s}){\bf x}_{i}\\[2pt]
{\bf 0}\\
\vdots\\[2pt]
{\bf 0}
\end{array}\right]=\left[\begin{array}{c}
{\bf 0}\\[2pt]
\hline \\[-8pt]
{\bf 0}\\
\vdots\\[2pt]
{\bf 0}\\
\lambda{\bf x}_{i}\\[2pt]
{\bf 0}\\
\vdots\\[2pt]
{\bf 0}
\end{array}\right].
\]
Clearly, this is equivalent to $A(K_{s}){\bf x}_{i}=\lambda{\bf x}_{i}$,
and ${\bf x}_{i}\perp{\bf 1}_{s}$ for $i=1,\ldots,\eta$. This yields
another $\eta(s-1)$ linearly independent eigenvectors associated
with the eigenvalue $\lambda=-1$.

Now, consider vectors of the form 
\begin{equation}
\left[\begin{array}{c}
{\bf 0}\\[2pt]
\hline \\[-8pt]
\alpha_{1}{\bf 1}_{s}\\
\vdots\\[2pt]
\alpha_{\eta}{\bf 1}_{s}
\end{array}\right],\label{eq:ev3}
\end{equation}
with the $\alpha_{i}$ not all zero. Then the equation $A{\bf x}=\lambda{\bf x}$
becomes 
\[
\left[\begin{array}{c}
\sum_{i=1}^{\eta}\alpha_{i}({\bf 1}_{s}^{T}{\bf 1}_{s}){\bf 1}_{c}\\[2pt]
\hline \\[-8pt]
\alpha_{1}A(K_{s}){\bf 1}_{s}\\
\vdots\\[2pt]
\alpha_{\eta}A(K_{s}){\bf 1}_{s}
\end{array}\right]=\left[\begin{array}{c}
{\bf 0}\\[2pt]
\hline \\[-8pt]
\lambda\alpha_{1}{\bf 1}_{s}\\
\vdots\\[2pt]
\lambda\alpha_{\eta}{\bf 1}_{s}
\end{array}\right].
\]
Since at least one of the $\alpha_{i}$ must be nonzero, this is equivalent
to 
\[
A(K_{s}){\bf 1}_{s}=\lambda{\bf 1}_{s},\qquad\alpha_{1}+\cdots+\alpha_{\eta}=0.
\]
The first of these two conditions implies that $\lambda=s-1$, and
there are exactly $\eta-1$ linearly independent eigenvectors of the
form (\ref{eq:ev3}) since the solution space of the homogeneous linear
equation $\alpha_{1}+\cdots+\alpha_{\eta}=0$ has dimension $\eta-1$.

To determine the two remaining eigenvectors and corresponding eigenvalues,
let 
\begin{equation}
{\bf x}=\left[\begin{array}{c}
{\bf 1}_{c}\\[2pt]
\hline \\[-8pt]
\beta{\bf 1}_{s}\\
\vdots\\[2pt]
\beta{\bf 1}_{s}
\end{array}\right],
\qquad\beta\ne0.\label{eq:ev4}
\end{equation}
Then the equation $A{\bf x}=\lambda{\bf x}$ becomes 
\[
\left[\begin{array}{c}
(c-1){\bf 1}_{c}+\eta\beta s\\[2pt]
\hline \\[-8pt]
c{\bf 1}_{s}+\beta(s-1){\bf 1}_{s}\\
\vdots\\[2pt]
c{\bf 1}_{s}+\beta(s-1){\bf 1}_{s}
\end{array}\right]=\left[\begin{array}{c}
\lambda{\bf 1}_{c}\\[2pt]
\hline \\[-8pt]
\lambda\beta{\bf 1}_{s}\\
\vdots\\[2pt]
\lambda\beta{\bf 1}_{s}
\end{array}\right],
\]
which reduces to the two conditions 
\begin{equation}
c-1+\eta\beta s=\lambda,\qquad c+\beta(s-1)=\lambda\beta.\label{eq:cond}
\end{equation}
Conditions (\ref{eq:cond}) yield 
\begin{equation}
\beta=\frac{\lambda-c+1}{\eta s},\qquad\beta=\frac{c}{\lambda-s+1}\,,\label{eq:cond1}
\end{equation}
hence 
\[
\frac{\lambda-c+1}{\eta s}=\frac{c}{\lambda-s+1}
\]
which is easily seen to become equation (\ref{eq:quadratic}) upon
multiplication of both sides by $\eta s(\lambda-s+1)$ and rearranging
terms. Solving for $\lambda$ yields the two remaining (simple) eigenvalues
of $A$: 
\begin{equation}
\lambda_{\pm}=\frac{1}{2}\left[c+s-2\pm\sqrt{(c-s)^{2}+4\eta cs}\,\right].\label{eq:roots}
\end{equation}
These two roots yield two distinct values $\beta_{\pm}$ of $\beta$
via (\ref{eq:cond1}) and thus two distinct eigenvectors of $A$ of
the form (\ref{eq:ev4}). It is obvious that these two eigenvectors
are linearly independent (indeed, mutually orthogonal). \end{proof}
\begin{rem}
This theorem generalizes the spectral analysis of windmill graphs
given in \cite{clustering divergence} to core-satellite graphs. 
\end{rem}

\begin{rem}
We emphasize that the proof of the previous theorem explicitly reveals
the structure of the eigenvectors of $A$; see equations (\ref{eq:ev1}),
(\ref{eq:ev2}), (\ref{eq:ev3}) and (\ref{eq:ev4}), together with
the fact that the eigenvectors associated with the complete graphs
are known. 
\end{rem}
We also point out another consequence of Theorem \ref{thm:spectrum}.
\begin{thm}
\label{thm:extreme} The two eigenvalues under item 3 in Theorem \ref{thm:spectrum}
are the extreme eigenvalues of $A$. In particular, the Perron eigenvalue
(spectral radius) $\rho(A)$ of $A$ is given by 
\[
\lambda_{+}=\frac{1}{2}\left[c+s-2+\sqrt{(c-s)^{2}+4\eta cs}\,\right]\,,
\]
while the leftmost eigenvalue of $A$ is given by 
\[
\lambda_{-}=\frac{1}{2}\left[c+s-2-\sqrt{(c-s)^{2}+4\eta cs}\,\right]\,.
\]

\end{thm}
Moreover, $\lambda_{-}$ is strictly less than the other eigenvalues
of $A$. 
\begin{proof}
We have 
\[
\lambda_{+}=\frac{c+s-2+\sqrt{(c-s)^{2}+4\eta cs}}{2}>\frac{c+s-2+\sqrt{(c-s)^{2}+4cs}}{2}=\frac{c+s-2}{2}+\frac{\sqrt{(c+s)^{2}}}{2}=\frac{c+s}{2}-1+\frac{c+s}{2}=c+s-1>s-1\,,
\]
proving that $\lambda_{+}$ exceeds the other eigenvalues of $A$
in magnitude.

Similarly, 
\[
\lambda_{-}=\frac{c+s-2-\sqrt{(c-s)^{2}+4\eta cs}}{2}<\frac{c+s-2}{2}-\frac{\sqrt{(c-s)^{2}+4cs}}{2}=\frac{c+s}{2}-1-\frac{\sqrt{(c+s)^{2}}}{2}=-1\,,
\]
proving that $\lambda_{-}$ is strictly less than all the other eigenvalues. \end{proof}

\begin{rem}
It follows that the principal eigenvector of $A$ is of the form (\ref{eq:ev4}).
Since this eigenvector is constant on each of the cliques forming
$\Theta(c,s,\eta)$, the eigenvector centrality \cite{EstradaBook}
of each node is the same for all the nodes in each clique, as one
would expect. Furthermore, for $\eta>1$ the nodes in the centrale
core have higher centrality than the nodes in the satellite cliques,
which is also to be expected. In other words, in (\ref{eq:ev4}) we
have that $0<\beta<1$. This can be shown using conditions (\ref{eq:cond1})
together with the computed expressions for $\lambda_{\pm}$. See also
Remark \ref{rm:centr} below. 
\end{rem}

\begin{rem} \label{rm:four}
It also follows that $\Theta(c,s,\eta)$ has exactly
four distinct eigenvalues. Moreover, it is clear from the expression
given for $\lambda_{+}$ that it must go to infinity as the graph
size $n=c+\eta s$ goes to infinity, i.e., if at least one of the
parameters $c$, $s$, or $\eta$ tends to infinity. It follows that
the {\em infection threshold}, defined as $\tau=1/\rho(A)$, vanishes
as the graph size grows, showing that core-satellite graphs of even
modest size offer almost no resistance to the spreading of infections,
gossip, rumors, etc. This is of course not surprising since the diameter
of a core-satellite graph is only 2. We refer to \cite{Dynamics_1}
for additional discussion of infection spreading in networks. See
also \cite{clustering divergence} for a discussion of the special
case of windmill graphs. 
\end{rem}

\section{Generalizing core-satellite graphs }

A more realistic scenario for modeling real-world complex networks
is to consider graphs in which not all satellite cliques are identical.
Consequently, we consider the generalization of core-satellite graphs to the case
where the satellites cliques are not restricted to have all the same
number of nodes. 

To this end,
let $t\ge1$ and consider a graph consisting of a core clique
with $c$ nodes and $\eta=\eta_{1}+\eta_{2}+\cdots+\eta_{t}$ satellite
cliques where $\eta_{1}$ cliques have $s_{1}$ nodes, $\eta_{2}$
cliques have $s_{2}$ nodes, ... , and $\eta_{t}$ cliques have $s_{t}$
nodes, where $s_{i}\ne s_{j}$ for $i\ne j$. Here $s_{i}\ge1$ and
$\eta_{i}\ge1$ for all $i=1,\ldots,t$.  As before, we assume that
each node in every satellite clique is connected to each node in the
core clique.

We introduce the following notation: letting 
\[
{\bf s}=(s_{1},s_{2},\ldots,s_{t}),\qquad{\boldsymbol{\eta}}=(\eta_{1},\eta_{2},\ldots,\eta_{t})\,,
\]
we denote with $\Theta(c,{\bf s},{\boldsymbol{\eta}})$ the {\em
generalized core-satellite graph} just decribed. Although this notation
is very convenient for expressing our mathematical results, it can
be very cumbersome for illustrative examples. For that last purpose
we will denote the number of cliques of a given size by an integer,
using a subindex for the size of the clique as illustrated in the
Figure \ref{Generalized CS}. Note that such a graph has $n=c+\sum_{i=1}^{t}\eta_{i}s_{i}$
nodes. Core-satellite graphs $\Theta(c,s,\eta)$ are obtained as special
cases for $t=1$. At the other end of the spectrum we have the generalized
core-satellite graphs in which all the satellite cliques have a different
number of nodes, corresponding to the case $t=\eta$, $\eta_{i}=1$
for all $i=1,\ldots,t$.

\begin{figure}
\begin{centering}
\includegraphics[width=0.65\textwidth]{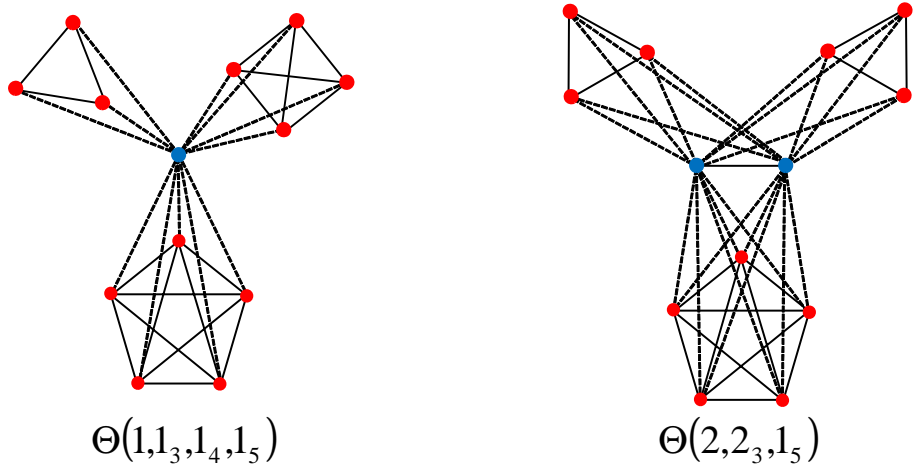} 
\par\end{centering}
\protect\protect\caption{Examples of generalized core-satellite graphs. Nodes in the core are
drawn in blue and those in the satellites in red. In this notation
for instance $1_{3},1_{4},1_{5}$ indicates that there is one clique
of size 3, one clique of size 4 and one clique of size 5.}
\label{Generalized CS} 
\end{figure}

We first investigate the phenomenon of clustering divergence and the
degree assortativity of these graphs. Unfortunately, it is difficult to extend
the analytic approach used in section \ref{sec:Gen} to the case of generalized 
core-satellite graphs. Therefore, we perform numerical experiments instead.
 Our computations indicate that
the clustering divergence phenomenon observed for the simple core-satellite
graphs also occurs in the generalized case. These graphs also display degree 
disassortativity like their simpler
analogues. In Figure \ref{properties GCS} we illustrate the divergence
of the clustering coefficients of these graphs and the increase of
their disassortativity as the number of satellite cliques in the graph increases.
For the sake of brevity, we only report here results for three kinds of generalized
core-satellite graphs having cores with 3, 5 and 10 nodes (we remark
that in real-world networks it is rare to find cliques od size larger
than 10 \cite{EstradaBook}). Then, we consider satellites of size
3, 5 and 7 and generate graphs $\Theta\left(c,\left(3,5,7\right),\left(p,p,p\right)\right)$
with $c\in \left\{ 3,5,10\right\} $ and $1\leq p\leq100$. As can be
seen in Fig. \ref{properties GCS}, as the number of satellites increases
the average Watts-Strogatz clustering coefficient approaches the asymptotic value 1,
while the transitivity index drops to zero. The increase of the Watts-Strogatz
index is slower for the graphs with bigger core as expected by the
fact that here more nodes with low clustering exists, namely, those
in the core. However, the assortativity coefficient of these graphs
decays much faster than that for graphs with smaller cores. In all
cases, the assortativity coefficient is negative, similarly to what
happens with the simple core-periphery graphs.

\begin{figure}
\begin{centering}
\includegraphics[width=0.65\textwidth]{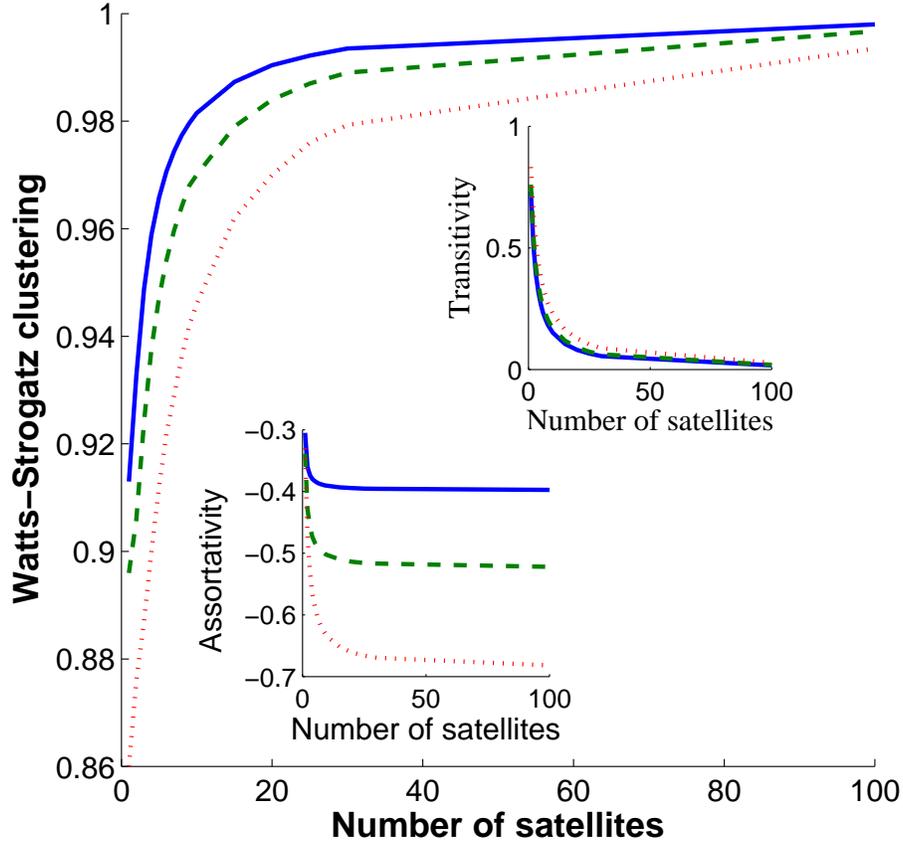}
\par\end{centering}

\protect\caption{Change in the average Watts-Strogatz clustering coefficient with the
increase of the number of satellite cliques for generalized core-satellite
graphs. The graphs have cores with size 3 (solid blue line), 5 (broken
green line) and 10 (dotted red line). Three sizes are used for the
satellites: 3, 5 and 7. The number of satellites of each type is the
same and varies between 1 and 100, i.e., we generate graphs 
$\Theta\left(c,\left(3,5,7\right),\left(p,p,p\right)\right)$
with $c\in \left\{ 3,5,10\right\} $ and $1\leq p\leq100$. In the insets
the plots for the global clustering (transitivity) and the assortativity
coefficient are shown for the same graphs. }

\label{properties GCS}
\end{figure}

Next, we study the 
spectral properties of the generalized core-satellite
graphs. 
Let us define the $c\times(\eta_{i}s_{i})$ matrices 
\[
B_{i}=\underbrace{\left[{\bf 1}_{c}{\bf 1}_{s_{i}}^{T}\,\,\,\,\,{\bf 1}_{c}{\bf 1}_{s_{i}}^{T}\,\,\,\,\,\cdots\,\,\,\,\,{\bf 1}_{c}{\bf 1}_{s_{i}}^{T}\right]}_{\eta_{i}{\rm \;times}}\,,\qquad i=1,\ldots,t\,,
\]
and the $(\eta_{i}s_{i})\times(\eta_{i}s_{i})$ block diagonal matrices
\[
A_{i}=\underbrace{A(K_{s_{i}})\oplus\cdots\oplus A(K_{s_{i}})}_{\eta_{i}{\rm \;times}}=\left[\begin{array}{cccc}
A(K_{s_{i}})\\
 & A(K_{s_{i}})\\
 &  & \ddots\\
 &  &  & A(K_{s_{i}})
\end{array}\right],\qquad i=1,\ldots,t
\]
(with $\eta_{i}$ identical diagonal blocks). Then, for a suitable
ordering of the nodes, the adjacency matrix of $\Theta(c,{\bf s},{\boldsymbol{\eta}})$,
can be written in block form as 
\begin{equation}
A=\left[\begin{array}{c|cccc}
A(K_{c}) & B_{1} & B_{2} & \cdots & B_{t}\\[2pt]
\hline \\[-5pt]
B_{1}^{T} & A_{1}\\[2pt]
B_{2}^{T} &  & A_{2}\\
\vdots &  &  & \ddots\\[2pt]
B_{t}^{T} &  &  &  & A_{t}
\end{array}\right]\,.\label{eq:adj2}
\end{equation}
This matrix is $n\times n$ with $n=c+\sum_{i=1}^{t}\eta_{i}s_{i}$.
We have the following generalization of Theorem \ref{thm:spectrum}. 
\begin{thm}
\label{thm:spectrum2} The spectrum of the generalized core-satellite
graph $\Theta(c,{\bf s},{\boldsymbol{\eta}})$ consists of 
\begin{enumerate}
\item The eigenvalue $\lambda=-1$ with multiplicity $c+\sum_{i=1}^{t}\eta_{i}(s_{i}-1)-1$; 
\item The eigenvalue $s_{i}-1$ with multiplicity $\eta_{i}-1$, for all
$i=1,\ldots,t$ with $\eta_{i}>1$; 
\item The roots of the following algebraic equation of degree $t+1$: 
\begin{equation}
(\lambda-c+1)\prod_{i=1}^{t}(\lambda-s_{i}+1)=c\,\sum_{i=1}^{t}\eta_{i}s_{i}\prod_{j\ne i}(\lambda-s_{j}+1)\,,\label{eq:algebraic}
\end{equation}
each of multiplicity one. 
\end{enumerate}
\end{thm}
\begin{proof}
It is straightforward to verify that $A$ has $c-1$ linearly independent
eigenvectors of the form (\ref{eq:ev1}) corresponding to the eigenvalue
$\lambda=-1$. Likewise, it is easy to check that $A$ has $\sum_{i=1}^{t}(s_{i}-1)\eta_{i}$
linearly independent eigenvectors of the form (\ref{eq:ev2}) with
$s=s_{i}$, also corresponding to the eigenvalue $\lambda=-1$. Hence,
$A$ has (at least) $c+\sum_{i=1}^{t}\eta_{i}(s_{i}-1)-1=c+\sum_{i=1}^{t}\eta_{i}s_{i}-\eta-1$
linearly independent eigenvectors corresponding to the eigenvalue
$\lambda=-1$; since the total number of nodes is $n=c+\sum_{i=1}^{t}\eta_{i}s_{i}$,
there remain $\eta+1=\sum_{i=1}^{t}\eta_{i}+1$ eigenvalues to account
for.

Let us first assume that $\eta_{i}>1$ for some $i$. Consider a nonzero
vector of the form 
\begin{equation}
{\bf x}=\left[\begin{array}{c}
{\bf 0}\\
\hline {\bf 0}\\
\vdots\\[2pt]
{\bf 0}\\
\hline \alpha_{1}{\bf 1}_{s_{i}}\\
\alpha_{2}{\bf 1}_{s_{i}}\\
\vdots\\[2pt]
\alpha_{\eta_{i}}{\bf 1}_{s_{i}}\\
\hline {\bf 0}\\
\vdots\\[2pt]
{\bf 0}
\end{array}\right]\,,\qquad\alpha_{1}+\cdots+\alpha_{\eta_{i}}=0\,,\label{eq:ev5}
\end{equation}
where the (possible) nonzeros in ${\bf x}$ occur in the positions
that correspond to the diagonal block $A_{i}$ in $A$. Then one can
easily check that ${\bf x}$ is eigenvector of $A$ associated to
the eigenvalue $\lambda=s_{i}-1$, and from the condition on the $\alpha_{i}$
we immediately deduce that there are $\eta_{i}-1$ linear independent
vectors of the form (\ref{eq:ev5}).

There remain $t+1$ eigenvalues to account for, where $t$ is between
$1$ and $\eta$. When $t=1$ all the satellite cliques are identical
and we are back under the assumptions of Theorem \ref{thm:spectrum},
hence the remaining two eigenvalues are the roots of the quadratic
equation (\ref{eq:quadratic}), to which (\ref{eq:algebraic}) reduces
for $t=1$.

Assume $t>1$ and consider a vector of the form 
\begin{equation}
{\bf x}=\left[\begin{array}{c}
{\bf 1}_{c}\\[2pt]
\hline \\[-8pt]
\beta_{1}{\bf 1}_{s_{1}\eta_{1}}\\
\vdots\\[2pt]
\beta_{t}{\bf 1}_{s_{t}\eta_{t}}
\end{array}\right].\label{eq:ev7}
\end{equation}
We note that vectors of this form are automatically orthogonal to
all eigenvectors of the form (\ref{eq:ev1}) (with ${\bf x}_{1}\perp{\bf 1}_{c}$),
to all eigenvectors of the form (\ref{eq:ev2}) (with ${\bf x}_{i}\perp{\bf 1}_{s_{i}}$),
and to all eigenvectors of the form (\ref{eq:ev5}). Also, the span
${\cal S}$ of all vectors of the form (\ref{eq:ev7}) has dimension
$t+1$. Hence, it must be the invariant subspace of $A$ spanned by
the eigenvectors associated with the remaining $t+1$ eigenvalues
of $A$. Finally, we observe that any vector in ${\cal S}$ is also
of the form (\ref{eq:ev7}) up to a scalar multiple.

For such vectors, the equation $A{\bf x}=\lambda{\bf x}$ becomes
\begin{equation}
\left[\begin{array}{c}
(c-1){\bf 1}_{c}+\sum_{i=1}^{t}\beta_{i}\eta_{i}s_{i}{\bf 1}_{c}\\[2pt]
\hline \\[-8pt]
c{\bf 1}_{s_{1}\eta_{1}}+\beta_{1}(s_{1}-1){\bf 1}_{s_{1}\eta_{1}}\\
\vdots\\[2pt]
c{\bf 1}_{s_{t}\eta_{t}}+\beta_{t}(s_{t}-1){\bf 1}_{s_{t}\eta_{t}}
\end{array}\right]=\left[\begin{array}{c}
\lambda{\bf 1}_{c}\\[2pt]
\hline \\[-8pt]
\lambda\beta_{1}{\bf 1}_{s_{1}\eta_{1}}\\
\vdots\\[2pt]
\lambda\beta_{t}{\bf 1}_{s_{t}\eta_{t}}
\end{array}\right].\label{cond:7}
\end{equation}
Conditions (\ref{cond:7}) can be rewritten in the form 
\begin{equation}
c-1+\sum_{i=1}^{t}\beta_{i}\eta_{i}s_{i}=\lambda\,,\label{cond:8}
\end{equation}
together with 
\begin{equation}
c+\beta_{i}(s_{i}-1)=\lambda\beta_{i}\,,\qquad i=1,\ldots,t.\label{cond:9}
\end{equation}
Rearranging conditions (\ref{cond:9}) we obtain 
\begin{equation}
\beta_{i}=\frac{c}{\lambda-s_{i}+1}\,,\qquad i=1,\ldots,t.\label{cond:10}
\end{equation}
Substituting (\ref{cond:10}) into (\ref{cond:8}) we obtain 
\[
c-1+c\,\sum_{i=1}^{t}\frac{\eta_{i}s_{i}}{\lambda-s_{i}+1}=\lambda\,,
\]
or, equivalently, 
\[
\sum_{i=1}^{t}\frac{\eta_{i}s_{i}}{\lambda-s_{i}+1}=\frac{\lambda-c+1}{c}\,.
\]
The left-hand side of the last equation can be rewritten as 
\[
\frac{\sum_{i=1}^{t}\eta_{i}s_{i}\prod_{j\ne i}(\lambda-s_{j}+1)}{\prod_{i=1}^{t}(\lambda-s_{i}+1)}\,,
\]
finally leading to the equation 
\[
(\lambda-c+1)\prod_{i=1}^{t}(\lambda-s_{i}+1)=c\sum_{i=1}^{t}\eta_{i}s_{i}\prod_{j\ne i}(\lambda-s_{j}+1)\,,
\]
which is precisely (\ref{eq:algebraic}). Next, we will show that
the $t+1$ roots of this equation are all distinct, therefore each
root $\lambda_{i}$ yields a distinct value of $\beta_{i}$ via (\ref{cond:10})
and therefore a distinct eigenvector (\ref{eq:ev7}), thus completing
the proof. For suppose that there is a root $\bar{\lambda}$ of (\ref{eq:algebraic})
of multiplicity at least two. Then the vector 
\[
{\bf x}=\left[\begin{array}{c}
{\bf 1}_{c}\\[2pt]
\hline \\[-8pt]
\bar{\beta}_{1}{\bf 1}_{s_{1}\eta_{1}}\\
\vdots\\[2pt]
\bar{\beta}_{t}{\bf 1}_{s_{t}\eta_{t}}
\end{array}\right],\qquad\bar{\beta}_{i}=\frac{c}{\bar{\lambda}-s_{i}+1}\,,\qquad i=1,\ldots,t,
\]
is an eigenvector of $A$ associated with $\bar{\lambda}$. But then
there must be another eigenvector ${\bf y}$ of $A$ associated with
the same eigenvalue $\bar{\lambda}$ and orthogonal to ${\bf x}$.
As argued above, this eigenvector ${\bf y}$ must lie in the subspace
${\cal S}$, and therefore be again of the form (\ref{eq:ev5}). But
when we impose ${\bf y}$ to be of the form (\ref{eq:ev5}) and to
satisfy $A{\bf y}=\bar{\lambda}{\bf y}$ we are forced to conclude
that ${\bf y}={\bf x}$, contradicting the orthogonality requirement.
Hence, all roots of (\ref{eq:algebraic}) must be simple. In particular,
when all the cliques have different size ($\eta_{i}=1$ for all $i=1,\ldots,t$)
there are $\eta+1$ eigenvalues given by the roots of (\ref{eq:algebraic}),
and no eigenvalues of the form $\lambda=s_{i}-1$. This completes
the proof. \end{proof}
\begin{rem}
It is well known (Abel--Ruffini Theorem) that the solution by radicals
of the general algebraic equation of degree 5 or higher is impossible.
We do not know if the equation (\ref{eq:algebraic}) is solvable by
radicals for any values of $s_{i}$, $\eta_{i}$, and $t\ge4$. For
$t=2$ and $t=3$ equation (\ref{eq:algebraic}) can be solved by
radicals, but the roots are given by rather complicated expressions
(Cardano's and Ferrari's formulas). Hence, we do not attempt to explicitly
write down the roots of (\ref{eq:algebraic}). 
\end{rem}

\begin{rem}
The number of distinct eigenvalues of $\Theta(c,{\bf s},{\boldsymbol{\eta}})$
is given by 
\[
t+2+|\{i\,|\,\eta_{i}>1\}|\,,
\]
where $|X|$ is the cardinality of the set $X$. This quantity is
equal to $t+2$ (or $\eta+2$, equivalently) when all the $t$ satellite
cliques have different sizes. \end{rem}
\begin{thm}
\label{thm:bounds} The spectral radius (Perron eigenvalue) $\rho(A)$
of $A$ is given by the largest root of (\ref{eq:algebraic}) and
satisfies the bounds 
\begin{equation}
c-1+\max_{1\le i\le t}s_{i}<\rho(A)<c-1+\sum_{i=1}^{t}\eta_{i}s_{i}\,.\label{eq:bounds}
\end{equation}
\end{thm}
\begin{proof}
We begin by recalling the well known inequalities 
\begin{equation}
\min_{1\le i\le n}\,\sum_{j=1}^{n}a_{ij}<\rho(A)<\max_{1\le i\le n}\,\sum_{j=1}^{n}a_{ij},\label{eq:bound1}
\end{equation}
which hold for the spectral radius of any nonnegative irreducible
matrix with non-constant row sums \cite[Lemma 2.5]{Varga}. In our
case the upper bound is the maximum degree, which is attained by each
node in the core clique and is equal to $n-1=c-1+\sum_{i=1}^{t}\eta_{i}s_{i}$.
The lower bound in (\ref{eq:bound1}) yields $\rho(A)>c+\min_{1\le i\le t}s_{i}-1$,
which in general is not enough to conclude that $\rho(A)$ is given
by the largest root of (\ref{eq:algebraic}) by Theorem \ref{thm:spectrum2}.
However, we can improve on this lower bound as follows. Assume (without
loss of generality) that $s_{1}=\max_{1\le i\le t}s_{i}$ and consider
a vector of the form 
\[
{\bf x}=\left[\begin{array}{c}
{\bf 1}_{c}\\[2pt]
\hline \\[-8pt]
{\bf 1}_{s_{1}\eta_{1}}\\
{\bf 0}\\
\vdots\\[2pt]
{\bf 0}
\end{array}\right],
\]
then 
\[
A{\bf x}=\left[\begin{array}{c}
(c-1+s_{1}\eta_{1}){\bf 1}_{c}\\[2pt]
\hline \\[-8pt]
(c-1+s_{1}){\bf 1}_{s_{1}\eta_{1}}\\
c{\bf 1}_{s_{2}\eta_{2}}\\
\vdots\\[2pt]
c{\bf 1}_{s_{t}\eta_{t}}
\end{array}\right]\ge(c-1+s_{i}){\bf x}\,,\qquad\forall i=1,\ldots,t.
\]
Hence, we have found a nonnegative vector ${\bf x}$ such that $A{\bf x}\ge r{\bf x}$
for $r=c-1+\max_{1\le i\le t}s_{i}$. By a well known result (see
\cite[section 3]{MS90}) this implies that $\rho(A)\ge r$. 
Finally, it
is straightforward to check that $\lambda=c-1+s_{1}$ is not a root
of (\ref{eq:algebraic}), therefore it must be $\rho(A)>c-1+s_{1}=c-1+\max_{1\le i\le t}s_{i}$,
proving the lower bound in (\ref{eq:bounds}). \end{proof}
\begin{rem}
\label{rm:centr} It follows that the principal eigenvector of $A$
is of the form (\ref{eq:ev7}). Since this eigenvector is constant
on each of the cliques forming $\Theta(c,{\bf s},{\boldsymbol{\eta}})$,
the eigenvector centrality of each node is the same for all the nodes
within a given clique, as one would expect. Furthermore, the nodes
in the core clique have higher centrality than the nodes in any of
the satellite cliques (assuming of course there is more than one of
them), which is also to be expected. In other words, in (\ref{eq:ev5})
we have that $0<\beta_{i}<1$ for all $i=1,\ldots,t$. To see this,
note that the $\beta_{i}$ are given by 
\[
\beta_{i}=\frac{c}{\rho(A)-s_{i}+1}\qquad(i=1,\ldots,t),
\]
see (\ref{cond:10}). By Theorem \ref{thm:bounds} we have $\rho(A)>c-1+s_{i}$
for all $i=1,\ldots,t$, hence $\rho(A)-s_{i}+1>c$ and therefore
$0<\beta_{i}<1$ for all $i=1,\ldots,t$. \end{rem}
\begin{thm}
The spectral radius of a generalized core-satellite graph goes to
infinity as the graph size $n$ goes to infinity. \end{thm}
\begin{proof}
Let $\hat{A}$ be the adjacency matrix of the star graph with $n$
nodes, with the central node numbered first. The spectrum of $\hat{A}$
consists of the eigenvalue 0 of multiplicity $n-2$, plus the eigenvalues
$\pm\sqrt{n-1}$. In particular, $\rho(\hat{A})=\sqrt{n-1}$. If $A$
is the adjacency matrix of a generalized core-satellite graph with
$n$ nodes, it is straightforward that $A\ge\hat{A}$, where the inequality
is component-wise. It is well known \cite[page 520]{HJ} that this
implies $\rho(A)\ge\rho(\hat{A})$, hence we conclude that the spectral
radius of $A$ tends to infinity as the graph size tends to infinity. \end{proof}
\begin{rem}
This result implies that in any generalized core-satellite graph,
the infection threshold $\tau$ vanishes as the graph size grows,
see Remark \ref{rm:four}. Again, since the diameter of a generalized
core-satellite graph is constant and equal to 2, this is to be expected. 
\end{rem}

We conclude this section with the description of the spectrum of the
graph Laplacian $L$ of $\Theta(c,{\bf s},{\boldsymbol{\eta}})$.
Letting $\gamma=\sum_{i=1}^{t}\eta_{i}s_{i}$ and 
\[
L_{i}=\underbrace{L(K_{s_{i}})\oplus\cdots\oplus L(K_{s_{i}})}_{\eta_{i}{\rm \;times}},
\]
the graph Laplacian of $\Theta(c,{\bf s},{\boldsymbol{\eta}})$ can
be written as 
\begin{equation}
L=\left[\begin{array}{c|cccc}
L(K_{c})+\gamma I_{c} & -B_{1} & -B_{2} & \cdots & -B_{t}\\[2pt]
\hline \\[-5pt]
-B_{1}^{T} & L_{1}+cI_{\eta_{1}s_{1}}\\[2pt]
-B_{2}^{T} &  & L_{2}+cI_{\eta_{2}s_{2}}\\
\vdots &  &  & \ddots\\[2pt]
-B_{t}^{T} &  &  &  & L_{k}+cI_{\eta_{t}s_{t}}
\end{array}\right]\,.\label{eq:Lap2}
\end{equation}

We have the following result. 
\begin{thm}
\label{thm:Lspectrum2} The spectrum of the graph Laplacian associated
with the generalized core-satellite graph $\Theta(c,{\bf s},{\boldsymbol{\eta}})$
consists of 
\begin{enumerate}
\item The eigenvalue $\mu=c+\sum_{i=1}^{t}\eta_{i}s_{i}=n$ with multiplicity
$c$; 
\item The eigenvalue $\mu=c+s_{i}$ with multiplicity $\eta_{i}(s_{i}-1)$,
for $1\le i\le t$; 
\item The eigenvalue $\mu=c$ with multiplicity $\sum_{i-1}^{t}\eta_{i}-1=\eta-1$; 
\item The eigenvalue $\mu=0$ with multiplicity one. 
\end{enumerate}
\end{thm}
\begin{proof}
It is straightforward to verify that $L$ has $c-1$ linearly independent
eigenvectors of the form (\ref{eq:ev1}) with ${\bf x}_{1}\perp{\bf 1}_{c}$
associated to the eigenvalue $\mu=c+\sum_{i=1}^{t}\eta_{i}s_{i}$.
In addition, the vector 
\[
{\bf x}=\left[\begin{array}{c}
{\bf 1}_{c}\\[2pt]
\hline \\[-8pt]
\beta{\bf 1}_{\eta_{1}s_{1}}\\
\vdots\\[2pt]
\beta{\bf 1}_{s_{t}\eta_{t}}
\end{array}\right],\qquad\beta=-\frac{c}{\sum_{i=1}^{t}\eta_{i}s_{i}}\,,
\]
is also eigenvector of $L$ associated with the eigenvalue $\mu=c+\sum_{i=1}^{t}\eta_{i}s_{i}=n$.

Next, it is easily checked that every nonzero vector of the form (\ref{eq:ev2})
with ${\bf x}_{i}\perp{\bf 1}_{s_{i}}$ is an eigenvector of $L$
associated with the eigenvalue $\mu=c+s_{i}$, and there are exactly
$\eta_{i}(s_{i}-1)$ linearly independent eigenvectors of that form.

Consider now vectors of the form 
\begin{equation}
{\bf x}=\left[\begin{array}{c}
{\bf 0}\\[2pt]
\hline \\[-8pt]
{\bf x}_{1}\\
\vdots\\[2pt]
{\bf x}_{t}
\end{array}\right]\,, \qquad{\rm where}\qquad{\bf x}_{i}=\left[\begin{array}{c}
\alpha_{1i}{\bf 1}_{s_{i}}\\
\alpha_{2i}{\bf 1}_{s_{i}}\\
\vdots\\[2pt]
\alpha_{\eta_{i}i}{\bf 1}_{s_{i}}
\end{array}\right] \qquad(i=1,\ldots,t).\label{eq:ev8}
\end{equation}
Then the equality $L{\bf x}=\mu{\bf x}$ is satisfied for $\mu=c$
if and only if 
\begin{equation}
\sum_{i=1}^{t}s_{i}\sum_{j=1}^{\eta_{i}}\alpha_{ji}=0\,.\label{eq:cond10}
\end{equation}
Equation (\ref{eq:cond10}) is a homogeneous linear equation in the
unknowns $\alpha_{ji}$, with $1\le i\le t$ and $1\le j\le\eta_{i}$.
Since the number of unknowns is $\eta_{1}+\cdots+\eta_{t}=\eta$,
equation (\ref{eq:cond10}) admits $\eta-1$ linearly independent
solutions, each leading to a distinct eigenvector (\ref{eq:ev8})
of $L$ associated with $c$.

Finally, since $\Theta(c,{\bf s},{\boldsymbol{\eta}})$ is connected,
$L$ has the simple eigenvalue zero associated with the eigenvector
with constant entries. \end{proof}
\begin{rem}
The Laplacian eigenvalues of $\Theta(c,{\bf s},{\boldsymbol{\eta}})$
can also be obtained from general results about the Laplacian eigenvalues
of the join of two (or more) graphs and the fact that the Laplacian
eigenvalues of a complete graph with $p$ nodes are known (they are
$\mu=p$ with multiplicity $p-1$ and $\mu=0$ with multiplicity 1);
see, e.g., \cite[Theorem 2.20]{Merris94}. The above proof has the
advantage of being self-contained and of explicitly exhibiting the
eigenvectors of $L$. 
\end{rem}

\begin{rem}
It follows from Theorem \ref{thm:Lspectrum2} that the Laplacian of
a generalized core-satellite graph $\Theta(c,{\bf s},{\boldsymbol{\eta}})$,
with $t>1$, has exactly $t+3$ distinct eigenvalues. For a core-satellite
graph $\Theta(c,s,\eta)$ (where $t=1$ and $\eta_{1}=\eta>1$), the
number of distinct eigenvalues is four: $\mu=0$, $\mu=c$, $\mu=c+s$,
and $\mu=c+\eta s$ (see also Remark \ref{rm:four}). 
\end{rem}

\begin{rem}
We observe that (generalized) core-satellite graphs are {\em Laplacian-integral},
i.e., the Laplacian eigenvalues are all integers. Moreover, the {\em
algebraic connectivity} (smallest nonzero Laplacian eigenvalue) of
a generalized core-satellite graph is equal to $c$, the size of the
core clique, and is independent of the number or size of the satellite
cliques. This is to be expected, since generalized core-satellite
graphs can be disconnected by disconnecting the core clique. 
\end{rem}

\begin{rem}
The ratio between the smallest nonzero eigenvalue of $L$ and the
largest eigenvalue of $L$ is known as the {\em synchronization
index} of the network: $Q=\mu_{n-1}/\mu_{1}$. Theorem \ref{thm:Lspectrum2}
implies that 
\[
Q=\frac{c}{c+\sum_{i=1}^{t}\eta_{i}s_{i}}=\frac{c}{n}\,.
\]
In particular, for windmill graphs we recover the fact that $Q=1/n$,
as already observed in \cite{clustering divergence}. Therefore, {\em
any two core-satellite graphs of the same size $n$ and with the same
core clique $K_{c}$ have exactly the same synchronization index},
a somewhat surprising result. Moreover, if $c$ is fixed and $n$
grows, the synchronization index decreases. This implies that unless
the core size $c$ grows and the number and size of satellite cliques
remains constant (or bounded), large (generalized) core-satellite
graphs are bad synchronizers. We refer to \cite{Chen2008,Chen2015,BaraPec02}
for detailed discussions of network synchronizability. 
\end{rem}
\vspace{0.1in}

\section{Conclusions }

Real-world graphs have a great variety of structures and complexities.
Hence, the existence of different mathematical models---Erd\H{o}s-R\'enyi,
Barab\'asi-Albert, Watts-Strogatz, random geometric graphs, etc. However,
not all of the structural properties of real-world graphs are
captured by these models. One simple example is the local-global clustering
coefficient divergence. This structural effect is simply due to the fact
that in some networks the local clustering tends to the maximum while
the global one tends to the minimum when the size of the graphs grows
to infinity. Here we have proposed a general class of graphs for which
this clustering divergence is observed and can be studied analytically.
These graphs---here proposed to be called core-satellite graphs---are
characterized by a central core of nodes which are connected to a
few satellites which may be of the same or different sizes. In this
work we have investigated some general properties of these graphs,
e.g., clustering coefficients, assortativity coefficients, as well
as the eigenstructure of both the adjacency and the Laplacian matrices.
Core-satellite graphs can also be easily modified so as to model other
properties of complex networks, such as hierarchical structure \cite{hierarchical}.
All of these make core-satellite graphs a flexible model for certain
classes of real-world networks, opening some new possibilities for
the analytic modeling of these systems.\\

\textbf{\Large{}{}{}Acknowledgements}{\Large \par}

 EE thanks the Royal Society of London for a Wolfson Research Merit
Award. MB was supported in part by National Science Foundation grant
DMS-1418889.

\end{document}